\documentclass[a4paper]{article}
\usepackage[applemac]{inputenc}
\usepackage[english]{babel}

\usepackage{latexsym}
\usepackage{amssymb}
\usepackage{amsmath}
\usepackage{amsthm}
\usepackage{cite}    
\usepackage{ifthen}
\usepackage{xspace}
\usepackage{tikz}

\usepackage[colorlinks,final,citecolor=blue]{hyperref}

\usetikzlibrary{%
  arrows,%
  automata,%
  positioning,%
  decorations.pathmorphing,%
  decorations.pathreplacing,%
  shadows} 

\newcommand{\set}[2]{\left\{#1\mathrel{\left|\vphantom{#1}\vphantom{#2}\right.}#2\right\}}
\newcommand{\oneset}[1]{\left\{\mathinner{#1}\right\}}

\newcommand{\abs}[1]{\left|\mathinner{#1}\right|}

\newcommand{\N}{\mathbb{N}}

\newcommand{\cH}{\ensuremath{\mathcal{H}}\xspace}

\newcommand{\cL}{\ensuremath{\mathcal{L}}\xspace}

\newcommand{\cR}{\ensuremath{\mathcal{R}}\xspace}

\newcommand{\sfA}{\ensuremath{\mathsf{A}}\xspace}
\newcommand{\sfC}{\ensuremath{\mathsf{C}}\xspace}
\newcommand{\sfG}{\ensuremath{\mathsf{G}}\xspace}
\newcommand{\sfT}{\ensuremath{\mathsf{T}}\xspace}

\newcommand{\BAR}{\overline{\phantom{ii}}}

\newcommand{\smalloverline}[1]
{{\mspace{1mu}\overline{\mspace{-1mu}#1\mspace{-1mu}}\mspace{1mu}}}
\newcommand{\ov}[1]{\smalloverline{#1}\vphantom{#1}}
\newcommand{\ovc}[1]{\smalloverline{#1}} 

\newcommand{\NL}{\ensuremath{\mathrm{NL}}\xspace}

\newcommand{\HCS}{\ensuremath{\mathrm{HCS}}\xspace}

\newcommand{\e}{\varepsilon}

\newcommand{\eg}{e.\,g.,\xspace}

\newcommand\Hk{\cH_k}
\newcommand\Hks{\cH_k^{*}}
\newcommand\Ha{\cH_\alp}
\newcommand\Has{\cH_\alp^{*}}
\newcommand\RH{{\cR\!\cH}}
\newcommand\LH{{\cL\!\cH}}

\newcommand{\sse}{\subseteq}
\newcommand{\es}{\emptyset}

\renewcommand{\phi}{\varphi}

\newcommand{\alp}{\alpha}
\newcommand{\bet}{\beta}
\newcommand{\gam}{\gamma}

\newcommand{\sig}{\sigma}

\newcommand{\Sig}{\Sigma}

\newcommand{\gabag}{\gamma\alpha\beta\ov\alpha\ov\gamma}
\newcommand{\gaba}{\gamma\alpha\beta\ov\alpha}
\newcommand{\abag}{\alpha\beta\ov\alpha\ov\gamma}

\newcommand{\aSa}{\alp\Sig^*\alp^{-1}}
\newcommand{\aSSa}{\alpha \Sigma^* \cap \Sigma^* \ov\alp}

\newcommand{\Pa}{\mathrm{P}_{\alp}}
\newcommand{\Sa}{\mathrm{S}_{\ov\alp}}
\newcommand{\ia}{{\mathrm{ind}_{\alp}}}

\newcommand{\ra}{\rightarrow}

\theoremstyle{plain}
\newtheorem{theorem}{Theorem}[section]
\newtheorem{proposition}[theorem]{Proposition}
\newtheorem{lemma}[theorem]{Lemma}
\newtheorem{corollary}[theorem]{Corollary}

\begin{document}

\title{Iterated Hairpin Completions of Non-crossing Words
	\thanks{
		This research was supported
		by the Natural Sciences and Engineering
		Research Council of Canada Discovery Grant R2824A01 and Canada Research
		Chair Award to L.\,K., and 
		by the Funding Program for Next Generation
		World-Leading Researchers (NEXT Program) to Yasushi Okuno,
		the current supervisor of S.\,S.}}

\author{Lila Kari \and Steffen Kopecki \and Shinnosuke Seki}

\maketitle

\begin{abstract}
	Iterated hairpin completion is an operation on formal languages
	that is inspired by the hairpin formation in DNA biochemistry.
	Iterated hairpin completion of a word (or more precisely a singleton language)
	is always a context-sensitive
	language and for some words it is known to be non-context-free.
	However, it is unknown whether regularity of iterated hairpin completion of
	a given word is decidable.
	Also the question whether iterated hairpin completion of a word
	can be context-free but not regular was asked in literature.
	In this paper we investigate iterated hairpin completions of non-crossing words
	and, within this setting, we are able to answer both questions.
	For non-crossing words we prove that the regularity of iterated hairpin completions
	is decidable and that if iterated hairpin completion of a non-crossing word
	is not regular, then it is not context-free either.
\end{abstract}

\section{Introduction}

On an abstract level, a DNA single strand can be viewed as a word over the
four-letter alphabet $\oneset{\sfA, \sfC, \sfG, \sfT}$
where the letters represent the nucleobases adenine, cytosine, guanine, and thymine, respectively.
The {\em Watson-Crick complement} of $\sfA$ is $\sfT$ and the complement of $\sfC$ is $\sfG$.
Two complementary single strands of opposite orientation can bond to each other
and form a DNA double strand.
Throughout the paper, we use the bar-notation for complementary strands of opposite orientation.

In the same manner, a single strand 
can bond to itself if two of its substrands are complementary
and do not overlap with each other.
Such an intramolecular base pairing is called a {\em hairpin}.
We are especially interested in hairpins of single strands 
of the form $\sig = \gaba$.
Here, the substrand $\ov\alp$ can bond to the substrand $\alp$.
Then, by extension, a new single strand can be synthesized
which we call a {\em hairpin completion} of $\sig$, see Figure~\ref{fig:hairpin}.
In this situation we call the substrands that initiate the hairpin completion,
$\alp$ and $\ov\alp$, {\em primers}.

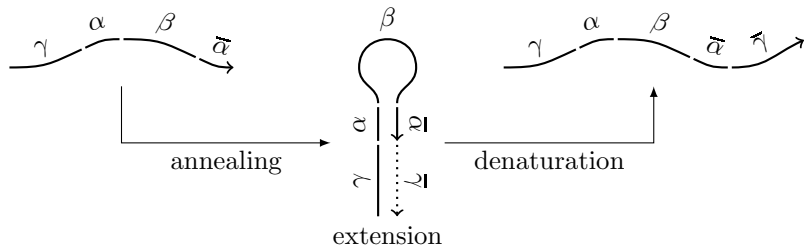
\begin{figure}[ht]
	\center
	\begin{tikzpicture}[text height=1.5ex,text depth=.25ex]
		\def\ang{25}
		\def\dist{3.5cm}
		
		\begin{scope}[above,sloped,every path/.style={thick,shorten <=.75pt,shorten >=.75pt}]
			\draw (-1.5,0) .. controls +(0:.5) and +(180+\ang:.5) .. node {$\gam$} (-.5,.25);
			\draw (-.5,.25) .. controls +(\ang:.25) and +(180:.25) .. node {$\alp$} (0,.375);
			\draw (0,.375) .. controls +(0:.5) and +(180-\ang:.5) .. node {$\bet$} (1,.125);
			\draw [->] (1,.125) .. controls +(-\ang:.25) and +(180:.25) 
				.. node {$\ov\alp$} (1.5,0);
		\end{scope}
		
		\begin{scope}[xshift=\dist, yshift=-1cm, every node/.style={above, sloped},
				every path/.style={thick,shorten <=.75pt,shorten >=.75pt},above,sloped]
			\draw (-.125,-1) -- node {$\gam$} (-.125,0);
			\draw (-.125,0) -- node {$\alp$} (-.125,.5);
			\draw (-.125,.5) .. controls +(90:.25) and +(-90:.25) .. (-.375,1)
				.. controls +(90:.5) and +(90:.5) .. node {$\bet$} (.375,1)
				.. controls +(-90:.25) and +(90:.25) .. (.125,.5);
			\draw [->] (.125,.5) -- node {$\ov\alp$} (.125,0);
			\draw [->,dotted] (.125,0) -- node {$\ov\gam$} (.125,-1);
		\end{scope}

		\begin{scope}[xshift=2*\dist -.5cm,above,sloped,
				every path/.style={thick,shorten <=.75pt,shorten >=.75pt}]
			\draw (-1.5,0) .. controls +(0:.5) and +(180+\ang:.5) .. node {$\gam$} (-.5,.25);
			\draw (-.5,.25) .. controls +(\ang:.25) and +(180:.25) .. node {$\alp$} (0,.375);
			\draw (0,.375) .. controls +(0:.5) and +(180-\ang:.5) .. node {$\bet$} (1,.125);
			\draw (1,.125) .. controls +(-\ang:.25) and +(180:.25) 
				.. node {$\ov\alp$} (1.5,0);
			\draw [->] (1.5,0) .. controls +(0:.5) and +(180+\ang:.5) 
				.. node {$\ov\gam$} (2.5,.375);
		\end{scope}
		
		\begin{scope}
			\draw [-latex] (0,-.25) -- (0,-1)  
				-- node [below] {annealing} (\dist -.75cm,-1);
			\draw [latex-] (2*\dist,-.25) -- (2*\dist,-1)  
				-- node [below] {denaturation} (\dist +.75cm,-1);
			\node at (\dist,-2.25) {extension};
		\end{scope}
	\end{tikzpicture}
	\caption{Hairpin completion of a DNA single strand}\label{fig:hairpin}
\end{figure}

In DNA computing hairpins and hairpin completions are often undesired by-products.
Therefore, sets of strands have been analyzed and designed that are unlikely to form
hairpins or lead to other {\em undesireable hybridization}, see
\cite{Adleman94,AritaKobayashi02,JonoskaMahalingam04,JoKeMa02,KaKoLoSoTh06,PaRoYo01}
and the references within.

However, there are DNA computational models that rely on hairpins,
\eg DNA RAM
\cite{KaYaOhYaHa08,TakinoueSuyama06,TakinoueSuyama04}
and Whiplash PCR
\cite{Hagiya97,Sakamoto98,Winfree98}.
For the Whiplash PCR consider a single strand just like in Figure~\ref{fig:hairpin},
but where the length of extension is controlled by {\em stopper sequences}.
Repeating this operation, DNA can be used to solve combinatorial problems
like the {\sc Hamiltonian path problem}.

Inspired by hairpins in biocomputing, the hairpin completion of a formal language has been
introduced by Cheptea, Mart\'\i{}n-Vide, and Mitrana in \cite{ChepteaMM06}.
In several papers hairpin completion and its iterated variant have been
investigated, see
\cite{DBLP:conf/cie/ManeaM07,ManeaMY09tcs,ManeaMM09,ManeaMY10,Manea10,DieKop10,Kopecki11}.
In this paper we consider iterated hairpin completions of singletons,
that is, informally speaking, iterated hairpin completions of words.
The class of iterated hairpin completions of singletons is denoted by \HCS.
It is known that every language in \HCS is decidable in \NL
(non-deterministic, logarithmic space) as \NL is closed under
iterated hairpin completion \cite{ChepteaMM06};
hence, \HCS is a proper subclass of the context-sensitive languages.
It is also known that \HCS contains regular as well as non-context-free languages \cite{Kopecki11}.
In the latter paper, two open problems have been stated:
\begin{enumerate}
	\item Is it decidable whether the iterated hairpin completion of a singleton is regular?
	\item Does a singleton exist whose iterated hairpin completion is context-free
		but not regular?
\end{enumerate}

We solve both questions for non-crossing words
(or rather, singletons containing a non-crossing word).
A word $w$ is said to be non-crossing if, for a given primer $\alp$, the
right-most occurrence of the factor $\alp$ in $w$ precedes the left-most occurrence
of the factor $\ov\alp$ in $w$, see Section~\ref{sec:non-crossing}.
We provide a necessary and sufficient condition for regularity of
iterated hairpin completion of a given non-crossing word (Theorem~\ref{thm:uequalv} and
Corollary~\ref{cor:condition})
and, since this condition is decidable, we answer the first question positively
(Corollary~\ref{cor:decidable}).
Furthermore, we show that iterated hairpin completion of a non-crossing word
is either regular or it is not context-free (Corollary~\ref{cor:noncf}).
Thus, we give a negative answer to the second question.

This paper is the continuation of the studies in \cite{KariKopeckiSeki}.

\section{Preliminaries}

We assume the reader to be familiar with the fundamental concepts of language theory,
see \cite{HU}.

Let $\Sigma$ be an alphabet, $\Sigma^*$ be the set of all words over $\Sigma$,
and for an integer $k \ge 0$, $\Sigma^k$ be the set of all words of length $k$ over $\Sigma$. 
The word of length 0 is called the empty word, denoted by
$\e$, and we let $\Sigma^+ = \Sigma^* \setminus \oneset\e$. 
A subset of $\Sigma^*$ is called a language over $\Sigma$. 
For a word $w \in \Sigma^*$, we employ the notation $w$
when we mean the word as well as the singleton language $\oneset w$
unless confusion arises. 

We equip $\Sigma$ with a function $\BAR \colon \Sigma \to \Sigma$
satisfying $\forall a \in \Sigma, \overline{\ov{a}} = a$;
such a function is called an {\em involution}. 
This involution $\BAR$ is naturally extended to words as:
for $a_1, \ldots, a_n \in \Sigma$, $\ov{a_1a_2 \cdots a_n} = \ov{a_n} \cdots \ov{a_2}\ov{a_1}$. 
For a word $w \in \Sigma^*$, we call $\ov{w}$ the {\em complement of $w$}, being inspired by this application. 
A word $w \in \Sigma^*$ is called a {\em pseudo-palindrome} if $w = \ov{w}$. 
For a language $L\sse \Sig^*$, we let $\ov{L} = \set{\ov{w}}{ w \in L }$. 

For words $u, w \in \Sigma^*$, if $w = xuy$ holds for some words
$x, y \in \Sigma^*$, then $u$ is called a {\em factor} of $w$;
a factor that is distinct from $w$ is said to be {\em proper}. 
If the equation holds with $x = \e$ ($y = \e$), then the factor $u$
is especially called a {\em prefix} (resp.\ a {\em suffix}) of $w$. 
The prefix relation can be regarded as a partial order $\le_p$ over $\Sig^*$
whereas the proper prefix relation can be regarded as a strict order $<_p$ over $\Sig^*$;
$u \le_p w$ means that $u$ is a prefix of $w$ and 
$u <_p w$ means that $u$ is a proper prefix of $w$.
Analogously, by $w \ge_s u$ (or $w >_s u$) we mean that $u$ is a suffix 
(resp.\ proper suffix) of $w$.
Note that $u\leq_p w$ if and only if $\ov w\geq_s \ov u$.
For a word $w \in \Sigma^*$ and a language $L \sse \Sigma^*$,
a factor $u$ of $w$ is {\em minimal with respect to $L$}
if $u \in L$ and none of the proper factors of $u$ is in $L$.

\subsection{Hairpin Completion}

Let $k$ be a constant that is assumed to be the length of a primer and
let $\alpha \in \Sigma^k$ be a primer. 
If a given word $w \in \Sigma^*$ can be written as
$\gaba$ for some $\gam,\bet \in \Sigma^*$,
then its {\em right hairpin completion} (with respect to $\alpha$) results
in the word $\gabag$. 
By $w \to_{\RH_\alpha} z$,
we mean that $z$ can be obtained from $w$ by right hairpin completion (with respect to $\alpha$).
The {\em left hairpin completion} is defined analogously as an operation to derive
$\gabag$ from $\abag$, and the relation $\to_{\LH_\alpha}$ is naturally introduced. 
We write $w\to_{\cH_\alp}z$ if $w\to_{\RH_\alp} z$ or $w \to_{\LH_\alp} z$.
By $\to_{\LH_\alp}^*$, $\to_{\RH_\alp}^*$, and $\to_{\cH_\alp}^*$
we denote the reflexive and transitive closure of $\to_{\LH_\alp}$,
$\to_{\RH_\alp}$, and $\to_{\cH_\alp}$, respectively. 
Whenever $\alp$ is clear from the context, we omit the subscript $\alp$ and write $\ra_\RH$, $\ra_\LH$,
or $\ra_\cH$, respectively.

For a language $L \subseteq \Sigma^*$, we define the set of words obtained by
hairpin completion from $L$, and the set of words obtained by iterated hairpin
completion from $L$, respectively, as follows: 
\begin{align*}
	\Ha(L) &= \set{z}{\exists w \in L, w \to_{\cH_\alpha} z}, &
	\Has(L) &= \set{z}{ \exists w \in L, w \to_{\cH_\alpha}^* z}. 
\end{align*}

In this paper the hairpin completion is always considered with respect to a fixed primer $\alp$.
However, in other literature the hairpin completion
is often considered with respect to the length $k$ of primers
instead of a specific primer $\alp$ and defined as
\begin{align*}
	&  \Hk(L) = \bigcup_{\alp\in\Sig^k} \Ha(L),
	&& \Hks(L) = \bigcup_{i\ge 0} \Hk^i(L).
\end{align*}

\section{Non-crossing Words and Their Properties}\label{sec:non-crossing}

In this section, we describe some structural properties of non-crossing
words and their iterated hairpin completions
and we introduce the notation of $\alp$-prefixes, $\alp$-suffixes, and $\alp$-indexes.

For a word $\alp$,
we say that $w$ is {\em non-$\alpha$-crossing} if the rightmost occurrence of
$\alpha$ precedes the leftmost occurrence of $\ov\alp$ on $w$
(yet these factors may overlap). 
If $\alpha$ is understood from the context, we simply say that $w$ is {\em non-crossing}. 
Otherwise, the word is {\em $\alpha$-crossing} or {\em crossing}.
The definition of a word $w$ being non-$\alpha$-crossing becomes useful in our work
only if $w \in \alp\Sig^*$ or $w\in \Sig^*\ov\alp$, and therefore, $\alp$ and $\ov\alp$ are primers;
actually, we will assume both.
The main purpose of this paper is to prove a necessary and sufficient condition
for the regularity of the iterated hairpin completion $\Has(w)$,
where $w\in \aSSa$ is non-$\alp$-crossing.

Note that if $w\in \aSSa$ and $\alpha = \ov{\alpha}$,
then either $w = \alp$ and $\Has(w) = \oneset w$
or $w$ can be considered crossing.
Thus, whenever we consider non-crossing words,
we assume that $\alpha \neq \ov{\alpha}$.

Any word obtained from a non-crossing word by hairpin completion is non-crossing. 
Though being easily confirmed, this closure property forms the foundation of our discussions in this paper. 

\begin{proposition}\label{prop:nonx_HC}
	For a non-crossing word $w\in\aSSa$, every word in $\Has(w)$ is non-crossing. 
\end{proposition}

Let us provide another characterization for a word $w\in\aSSa$ to be non-crossing.
With Proposition~\ref{prop:nonx_HC}, this characterization will
bring a unique factorization
of any word $z$ in $\Has(w)$ as $z = xwy$ for some words $x, y$
(Corollary~\ref{cor:nonx_initial_once}).

\begin{proposition}\label{prop:characterization}
	A word $w \in \alpha\Sigma^* \cap \Sigma^*\ov\alpha$ is non-crossing
	if and only if it contains exactly one factor $x$ which is minimal
	with respect to $\alpha\Sigma^* \cap \Sigma^*\ov\alpha$.
\end{proposition}

\begin{proof}
	The word $w$ contains at least one factor from $\alp\Sig^*\cap\Sig^*\ov\alp$
	and hence it contains at least one minimal factor, too.
	
	Suppose $w$ contains exactly one minimal factor $x$
	from $\alp\Sig^*\cap\Sig^*\ov\alp$.
	If the prefix $\alp$ of $x$ was proceeded in $w$ by another factor
	$\alp$, then this factor is again proceeded by $\ov\alp$ and we would find
	a second minimal factor $x'\in\alp\Sig^*\cap\Sig^*\ov\alp$.
	Hence, the prefix of $x$ is the rightmost occurrence of $\alp$ in $w$ and,
	symmetrically, the suffix of $x$ is the leftmost occurrence of $\ov\alp$
	in $w$.
	As the rightmost occurrence of $\alp$ precedes the leftmost occurrence
	of $\ov\alp$, the word $w$ is non-crossing. 
	
	Now suppose $w$ contains at least two minimal factors $x_1,x_2$
	from $\alp\Sig^*\cap\Sig^*\ov\alp$.
	If these factors overlap and the overlapped part is of length at least $k$,
	we may assume $x_1 = y_1y_2$ and $x_2 = y_2y_3$ for words $y_1,y_2,y_3\in\Sig^+$ where
	$y_2$, the overlapping part, is at least $k$.
	We see that $y_2 \in \alp\Sig^* \cap \Sig^*\ov\alp$
	and that $x_1, x_2$ were not minimal with respect to $\alp\Sig^* \cap \Sig^*\ov\alp$.
	Therefore, $x_1, x_2$ overlap in less than $k$ letters (or do not overlap at all)
	and $w$ can be considered crossing.
\end{proof}

\begin{corollary}\label{cor:nonx_initial_once}
	Let $w\in\aSSa$ be non-crossing.
	The factor $w$ occurs exactly once in every word from $\Has(w)$.
\end{corollary}

\subsection{\texorpdfstring{$\alp$-Prefixes and $\ovc\alp$-Suffixes}
	{alpha-Prefixes and alpha-Suffixes}}

Let $u, v, w$ be words.
We call $u$ an {\em $\alp$-prefix} of $w$ if $u\alp \leq_p w$.
This means, if $\ov\alp$ is a suffix of $w$, then the suffix can bond to
the factor $\alp$ which directly follows the prefix $u$
(unless they overlap with each other)
and $w\ov u$ can be obtained from $w$ by right hairpin completion.
By $\Pa(w)$, we denote the set of all $\alp$-prefixes of $w$.
Note that if $x, y \in \Pa(w)$ and $\abs{x} \leq \abs{y}$,
then $x\alp \leq_p y\alp$ and $x\in\Pa(y\alp)$.
Symmetrically, we call $\ov v$ an {\em $\ov\alp$-suffix} if $w\geq_s \ov\alp\ov v$
and $\Sa(w)$ is the set of all $\ov\alp$-suffixes of $w$.
If $\alp\leq_p w$ and $\abs w \geq \abs v + 2k$,
then $w \to_\LH vw$.
Note that $\Sa(w) = \ov{\Pa(\ov w)}$.
Therefore, all results we derive for $\alp$-prefixes also hold for the complements
of $\ov\alp$-suffixes.

When $m = \abs{\Pa(w)}$ and $n = \abs{\Sa(w)}$ for a word $w$, then
$w$ is called {\em $(m,n)$-$\alp$-word} (or simply {\em $(m,n)$-word}).
Throughout the paper, it will be convenient to let
$\Pa(w)= \oneset{u_0,\ldots,u_{m-1}}$ and $\Sa(w) = \oneset{\ov{v_0},\ldots,\ov{v_{n-1}}}$
where the words are ordered such that $u_0 <_p \cdots <_p u_{m-1}$ and
$\ov{v_{n-1}} >_s \cdots >_s \ov{v_0}$.
Note that $\Pa(u_i\alp) = \oneset{u_0,\ldots,u_i}$ for $0\leq i < m$
and $\ov{\Sa(\ov\alp\ov{v_j})} = \oneset{v_0,\ldots,v_j}$ for $0\leq j < n$.

Let us begin with a basic observation.

\begin{lemma}\label{lem:alpha}
	For a word $w \in \alp\Sig^* \cap \Sig^*\ov\alp$, the following statements hold:
	\begin{enumerate}
		\item For every $x \in \Pa(w) \cup \ov{\Sa(w)}$, we have $\alp \leq_p x\alp$.
		\item For every $x_1,\ldots,x_\ell \in\Pa(w) \cup \ov{\Sa(w)}$, we have 
			$\alp \leq_p x_\ell\cdots x_1\alp$.
	\end{enumerate}
\end{lemma}

\begin{proof}
	If $x\in \Pa(w)$, the first statement derives directly from the definition.
	If $x\in\ov{\Sa(w)}$, then $\ov\alp$ is a suffix of $\ov\alp \ov x$,
	whence $\alp \leq_p x\alp$.
	
	The second statement follows by induction on $\ell$ and the first statement.
\end{proof}

Consider $w \in \alp\Sig^* \cap \Sig^*\ov\alp$.
Note that this means $u_0 = \ov{v_0} = \e$.
It is easy to see that every word $z$ which belongs to $\Ha(w)$ has a
factorization $z = w\ov u$ for some $u\in\Pa(w)$ or
$z = v w$ for some $\ov v\in\Sa(w)$.
By the previous lemma we see that $z \in \alp\Sig^* \cap \Sig^*\ov\alp$ and by induction
$\Has(w) \sse \alp\Sig^* \cap \Sig^*\ov\alp$.

The next lemma tells if $w\in\aSSa$ is a non-crossing $(m,n)$-word
with $n\geq 2$, then the suffix $\ov\alp$ does not overlap with
any of the factors $\alp$ and, therefore, $w \to_\RH w\ov u$ for all $u\in\Pa(w)$.

\begin{lemma}\label{lem:length_nonoverlap_apref_casuf}
	Let $w\in\aSSa$ be a non-crossing $(m,n)$-word with $n\geq 2$
	and let $u_{m-1}$ be the longest $\alp$-prefix in $\Pa(w)$.
	Then $\abs{u_{m-1}}+2k \leq \abs w$ holds.
\end{lemma}

\begin{proof}
	Let $\Pa(w) = \oneset{u_0,\ldots,u_{m-1}}$ such that $u_0 <_p \cdots <_p u_{m-1}$
	and $\Sa(w) = \oneset{\ov{v_0},\ldots,\ov{v_{n-1}}}$
	such that $\ov{v_{n-1}} >_s \cdots >_s \ov{v_0}$ as usual.
	Suppose that the inequality $\abs{u_{m-1}}+\abs{v_{n-2}}+2k \le \abs w$ did not hold
	(note that this inequality is stronger than the one proposed in our claim).
	Being non-crossing, $w$ can be written as $w = u_{m-1}x\ov{v_{n-2}}$
	for some word $x\in \alp\Sig^*\cap \Sig^*\ov\alp$ with $\abs x \leq 2k$.
	Hence $x = \ov x$.
	Let $y$ be the nonempty word satisfying $\ov{v_{m-1}} = \ov y\ov{v_{m-2}}$.
	Since $w$ is non-crossing, $x\ov{v_{n-2}} \ge_s \ov{v_{n-1}}$ must hold,
	from which we have $x \ge_s \ov\alp \ov y$.
	Combining this with $x = \ov x$ enables us to find an $\alp$-prefix $u_{m-1}y$
	of $w$, but this would be longer than the longest $\alp$-prefix of $w$ --- a contradiction.
\end{proof}

Since the analogous argument is valid for left hairpin completion,
Lemma~\ref{lem:length_nonoverlap_apref_casuf} leads us to one important
corollary on non-crossing $(m, n)$-words for $m, n \ge 2$.

\begin{corollary}\label{cor:1st_step_mn_ge2}
	Let $w\in\aSSa$ be a non-crossing $(m,n)$-word with $m,n\ge 2$.
	\begin{equation*}
		\Ha(w) = \oneset w\cup w\ov{\Pa(w)} \cup \ov{\Sa(w)}w.
	\end{equation*}
\end{corollary}

Next, we concern the case when there is a prefix $u\in\Pa(w)$ and a suffix
$\ov{v}\in \Sa(w)$
such that $v\in \Pa(u\alp)^*$ and $u\in\ov{\Sa(\ov\alp\ov v)}^*$.

\begin{lemma}\label{lem:inclusions}
	Let $w\in\aSSa$ be an $(m,n)$-word.
	For $u\in\Pa(w)$ and $\ov v\in\Sa(w)$,
	\begin{enumerate}
		\item if $v\in \Pa(u\alp)^*$,
			then $\ov{\Sa(\ov\alp\ov v)} \sse \Pa(u\alp)^*$.
		\item if $v \in \Pa(u\alp)^*$ and $u\in\ov{\Sa(\ov\alp\ov v)}^*$,
			then $\Pa(u\alp)^* = \ov{\Sa(\ov\alp\ov v)}^*$.
			$\vphantom{\Big|}$ 
	\end{enumerate}
\end{lemma}

\begin{proof}
	Let $\Pa(w) = \oneset{u_0,\ldots,u_{m-1}}$ such that $u_0 <_p \cdots <_p u_{m-1}$ and
	$\Sa(w) = \oneset{\ov{v_0},\ldots,\ov{v_{n-1}}}$ such that $\ov{v_{n-1}} >_s \cdots >_s \ov{v_0}$.
	We may assume $u = u_s$ for some $0\leq s <m$.
	Recall that $\Pa(u\alp) = \oneset{u_0,\ldots,u_s}$.
	
	For the first statement,
	let $v = x_1\cdots x_\ell$ where $x_1,\ldots,x_\ell\in \oneset{u_0,\ldots, u_s}$.
	For $0\leq j \leq t$ we have $v_j\alp \leq v\alp$, hence,
	there is $1\leq i\leq \ell$ such that $v_j = x_1\cdots x_{i-1}y$
	where $y\leq_p x_{i}$ and $y\alp \leq_p x_{i}\cdots x_\ell\alp$.
	By Lemma~\ref{lem:alpha}, we see that $y\alp \leq_p x_i\alp \leq_p w$
	and hence $y\in \oneset{u_1,\ldots, u_s}$.
	Therefore, $v_j \in\oneset{u_1,\ldots, u_s}^*$.

	The second statement follows immediately by the symmetry between prefixes and complemented suffixes.
\end{proof}

Lemma~\ref{lem:theo} and Corollary~\ref{cor:sufone} in in Section~\ref{sec:main}
will describe the consequences for the iterated hairpin completion of a non-crossing word $w$,
if we find such a situation.

\subsection{\texorpdfstring{The $\alp$-Index}{The alpha-Index}}

The {\em $\alp$-index} of a word $x$ is the number of occurrences
of the factor $\alp$ in the word $x\alp$ except for the suffix.
Formally, we define a function $\ia\colon\Sig^*\to \N$ as
\begin{equation*}
	\ia(x) = \abs{\Pa(x\alp)}-1.
\end{equation*}

Recall that $\Pa(w) = \oneset{u_0,\ldots,u_{m-1}}$ such that $u_0 <_p \cdots <_p u_{m-1}$.
Note that, by this ordering, for all $0\leq i < m$ we have $\ia(u_i) = i$ and
if $x \in \Pa(w)$, then $x = u_{\ia(x)}$.
Symmetrically, if $\ov x \in \Sa(w)$, then $x = v_{\ia(x)}$.

Also note that for words $x,y$ with $\ia(x) > \ia(y)$ the
word $x$ cannot be a factor of $y$ as the positions of the factors $\alp$ cannot match.
Later we will consider the $\alp$-indices of words from $\aSa$
(Note that $x\in\aSa$ if and only if $\alp \leq_p x\alp$).
If $x\in\aSa$ and $y\in\Sig^*$, then $\ia(yx) = \ia(y) + \ia(x)$.
These observations lead to the following properties.

\begin{lemma} \label{lem:karl}
	Let $\Pa(w) = \oneset{u_0,\ldots,u_{m-1}}$ such that $u_0 <_p \cdots <_p u_{m-1}$,
	let $x\in\aSa$, and let $0 \leq j < m$.
	If $x$ is a suffix of $u_j$, then $u_j = u_{j-\ia(x)}x$.
	In particular, if $w\in\alp\Sig^*$ and $u_j \geq_s u_i $ for $0\leq i \leq j < m$, then
	$u_j = u_{j-i}u_i$.
\end{lemma}

\begin{proof}
	Let $y$ such that $u_j = yx$.
	As $\alp \leq_p x\alp$, we have $y\alp \leq_p u_j\alp \leq_p w$,
	and hence, $y\in \Pa(w)$.
	As $\ia(y) = \ia(u_j) - \ia(x)$,
	it follows that $y = u_{j-\ia(x)}$.
\end{proof}

\begin{lemma} \label{lem:siegbert}
	Let $\Pa(w) = \oneset{u_0,\ldots,u_{m-1}}$ such that $u_0 <_p \cdots <_p u_{m-1}$,
	let $1\leq i < m$, and let $x$ be a word with $\ia(x) \leq i$.
	\begin{equation*}
		x\in\oneset{u_1,\ldots,u_i}^* \quad\iff\quad
		x\in\oneset{u_1,\ldots,u_{\ia(x)}}^*.
	\end{equation*}
\end{lemma}

\begin{proof}
	The implication from right to left is plain.

	Conversely,
	let $x = y_1\cdots y_\ell$ with $y_1,\ldots,y_\ell\in \oneset{u_1,\ldots,u_i}$.
	As $\ia(x) \geq \ia(y_j)$ for all $1\leq j\leq \ell$,
	we see that, actually, $y_1,\ldots,y_\ell\in \oneset{u_1,\ldots,u_{\ia(x)}}$.
	Hence, $x\in \oneset{u_1,\ldots,u_{\ia(x)}}^*$, as desired.
\end{proof}

\section{Iterated Hairpin Completion of Non-crossing Words}
\label{sec:main}

Now we are prepared to prove a necessary and sufficient condition for the
regularity of $\Has(w)$, where $w\in\alp\Sig^*\cap\Sig^*\ov\alp$
is a non-crossing $(m,n)$-word
with $\Pa(w) = \oneset{u_0,\ldots,u_{m-1}}$ and $\Sa(w) = \oneset{\ov{v_0},\ldots,\ov{v_{n-1}}}$
which are ordered as in the previous section.
(Keep in mind that $u_0 = \ov{v_0} = \e$.)
By a result from \cite{KariKopeckiSeki}
it is enough to consider the case where $m,n \geq 2$
and in this case, by Corollary~\ref{cor:1st_step_mn_ge2},
\begin{equation*}
	\Ha(w) = \oneset w\cup w\oneset{\ov{u_1}, \ldots, \ov{u_{m-1}}}
		\cup \oneset{v_1, \ldots, v_{n-1}}w.
\end{equation*}

\begin{theorem}[See \cite{KariKopeckiSeki}]
	If $w\in\aSSa$ is a non-crossing $(m,n)$-word with $m=1$ or $n=1	$,
	then $\Has(w)$ is regular.
\end{theorem}

The next two lemmas lead to a first sufficient condition
for the regularity of $\Has(w)$, see Corollary~\ref{cor:sufone}.

\begin{lemma}\label{lem:otto}
	For non-crossing $w\in \aSSa$,
	\begin{equation*}
		\Has(w) \sse \left( \Pa(w) \cup \ov{\Sa(w)}\right)^* w
			\left( \ov{\Pa(w)} \cup \Sa(w) \right)^*.
	\end{equation*}
\end{lemma}

\begin{proof}
	For every word $z\in\Has(w)$, by definition,
	we find a series of words $w = w_0,w_1,\ldots,w_\ell = z$ for
	some $\ell\geq 0$ such that
	\begin{equation*}
		w_0 \ra_\cH w_1 \ra_\cH \cdots \ra_\cH w_\ell.
	\end{equation*}
	We prove $z = w_\ell\in\bigl( \Pa(w) \cup \ov{\Sa(w)}\bigr)^* w
	\bigl( \ov{\Pa(w)} \cup \Sa(w) \bigr)^*$ by induction on $\ell$.
	For $\ell = 0$ this is plain.
	Now we assume that any word which can be derived from $w$
	by at most $\ell-1$ hairpin completions is in this set.
	By induction hypothesis,
	\begin{equation*}
		w_{\ell-1}= x_s\cdots x_1 w \ov{y_1}\cdots \ov{y_t}
	\end{equation*}
	for some $s,t\geq 0$ and $x_1,\ldots,x_s,y_1,\ldots,y_s \in \Pa(w) \cup \ov{\Sa(w)}$.
	Hence, $\Pa(w_{\ell-1}) \sse \bigl( \Pa(w) \cup \ov{\Sa(w)}\bigr)^*$
	and $\Sa(w_{\ell-1}) \sse \bigl( \ov{\Pa(w)} \cup \Sa(w) \bigr)^*$.
	This proves our claim.
\end{proof}

\begin{lemma}\label{lem:theo}
	Let $w\in\aSSa$ be non-crossing.
	If for some $u\in\Pa(w)$ and $\ov v\in\Sa(w)$
	we have $v\in\Pa(u\alp)^*$ and $u\in\ov{\Sa(\ov\alp\ov v)}^*$, then
	\begin{equation*}
		\Pa(u\alp)^* w \ov{\Pa(u\alp)}^* = 
		\ov{\Sa(\ov\alp\ov v)}^* w \Sa(\ov\alp\ov v)^*
		\sse \Has(w).
	\end{equation*}
\end{lemma}

\begin{proof}
	Let $\Pa(w) = \oneset{u_0,\ldots,u_{m-1}}$ such that $u_0 <_p \cdots <_p u_{m-1}$ and
	$\Sa(w) = \oneset{\ov{v_0},\ldots,\ov{v_{n-1}}}$ such that $\ov{v_{n-1}} >_s \cdots >_s \ov{v_0}$.

	First we prove that for the special case when $u = v$,
	$\Pa(u\alp)^* w \ov{\Pa(u\alp)}^* \sse \Has(w)$.
	Afterwards, we will use this fact in order to prove our claim.
	
	Suppose $u = v\in\Pa(w)\cap\ov{\Sa(w)}$ and let $\ell = \ia(u)$.
	This implies $u_i = v_i$ for all $0\le i\le \ell$.
	Let
	\begin{equation*}
		z = x_s\cdots x_1 w \ov{y_1}\cdots \ov{y_t}
	\end{equation*}
	where $s,t\geq 0$ and $x_1,\ldots,x_s,y_1,\ldots, y_t\in \oneset{u_0,\ldots,u_\ell} = \Pa(u\alp)$.
	Let $s'$ and $t'$ be maximal such that $\ia(x_{s'}) = \ell$ and $\ia(y_{t'}) = \ell$;
	or $0$ if no such index exists.
	We let $w' = x_{s'}\cdots x_1 w \ov{y_1}\cdots\ov{y_{t'}}$.
	Note that $\ov\alp\ov{u_\ell}$ is a suffix of $w'$ and hence
	\begin{equation*}
		w \ra_\RH^* w \ov{y_1}\cdots\ov{y_{t'}} \ra_\LH^*
		x_{s'}\cdots x_1 w \ov{y_1}\cdots\ov{y_{t'}} = w'.
	\end{equation*}

	We have $w'\in\Has(w)$ and
	$z \in \oneset{u_0,\ldots,u_{\ell-1}}^* w'\oneset{\ov{u_0},\ldots,\ov{u_{\ell-1}}}^*$.
	We may continue inductively and conclude $z\in\Has(w)$.

	Now let $v\in\Pa(u\alp)^*$ and $u\in\ov{\Sa(\ov\alp\ov v)}^*$ as in the claim and
	let $\ell\leq \ia(u)$ be the minimal index such that $v \in \oneset{u_0,\ldots,u_\ell}^*$.
	Note that
	\begin{equation*}
		\Pa(u\alp)^* = \ov{\Sa(\ov\alp\ov v)}^* = \oneset{u_0,\ldots,u_\ell}^*,
	\end{equation*}
	by Lemma~\ref{lem:inclusions}.
	The minimality of $\ell$ yields $u_\ell \notin\oneset{u_0,\ldots,u_{\ell-1}}^*$
	and $u_\ell$ is a factor of $v$, which means that $\ell \leq \ia(v) < n$.
	We claim that $u_\ell = v_\ell$.
	Indeed, if we had $u_\ell \neq v_\ell$, then
	$u_\ell \in \oneset{v_0,\ldots, v_{\ell-1}}^* \sse \oneset{u_0,\ldots,u_{\ell-1}}^*$
	(see Lemma~\ref{lem:siegbert}).

	As $u_\ell = v_\ell$, our observations made for the special case apply
	and we may conclude
	\begin{equation*}
		\Pa(u_\ell\alp)^* w \ov{\Pa(u_\ell\alp)}^* = 
		\Pa(u\alp)^* w \ov{\Pa(u\alp)}^* = 
		\ov{\Sa(\ov\alp\ov v)}^* w \Sa(\ov\alp\ov v)^*
		\sse \Has(w).
	\end{equation*}
\end{proof}

Suppose that the longest $\alp$-prefix $u_{m-1}$ belongs to $\ov{\Sa(w)}^*$ and the longest
$\ov\alp$-suffix $\ov{v_{n-1}}$ belongs to $\ov{\Pa(w)}^*$.
By Lemma~\ref{lem:inclusions}, we see that
$\Pa(w)^*= \ov{\Sa(w)}^*$ and,
by Lemmas~\ref{lem:otto} and~\ref{lem:theo}, we infer
$\Has(w) = \Pa(w)^* w \ov{\Pa(w)}^*$.
(Note that $\Pa(w) = \Pa(u_{m-1}\alp)$.)

\begin{corollary}\label{cor:sufone}
	Let $w\in\aSSa$ be non-crossing, let $u_{m-1}$ be the longest $\alp$-prefix of $w$,
	and let $v_{n-1}$ be the longest $\ov\alp$-suffix of $w$.
	If $u_{m-1} \in \ov{\Sa(w)}^*$ and $v_{n-1}\in\Pa(w)^*$, then $\Has(w)$ is regular.
\end{corollary}

Our next result, Theorem~\ref{thm:uequalv}, shows
that if we  can state a necessary and sufficient condition for the 
special case where $u_1 = v_1$, then we can extend this condition to
the general case.
We need a preliminary lemma in order to prove Theorem~\ref{thm:uequalv}.

\begin{lemma}\label{lem:heinz}
	Let $w\in\aSSa$ be a non-crossing $(m,n)$-word with $m,n\geq 2$,
	let $\Pa(w) = \oneset{u_0,\ldots,u_{m-1}}$ such that $u_0 <_p \cdots <_p u_{m-1}$,
	and let $\Sa(w) = \oneset{\ov{v_0},\ldots,\ov{v_{v-1}}}$ such that
	$\ov{v_{n-1}} >_s \cdots >_s \ov{v_0}$.
	\begin{enumerate}
		\item If $u_1 = v_1$, then $\Has(w) \sse u_1\alp\Sig^* \cap \Sig^*\ov\alp\ov{u_1}$.
		\item If $u_1\neq v_1$, then $\Has(w\ov{u_i})\cap\Has(v_j w) = \es$
			for all $1\leq i < m$ and $1\leq j < n$.
		\item Let $1\leq i < j < m$.
			If $u_j >_s u_i$, then $\Has(w\ov{u_j}) \sse \Has(w\ov{u_i})$;
			otherwise, 
			$\Has(w\ov{u_i}) \cap \Has(w\ov{u_j}) = \es$.
	\end{enumerate}
\end{lemma}

\begin{proof}
	Consider a word $z$ with prefix $u_1\alp$ and suffix $\ov\alp \ov{u_1}$.
	It is easy to see that any word which is a hairpin completion of $z$ has prefix
	$u_1\alp$ and suffix $\ov\alp \ov{u_1}$, as well.
	Statement~1 follows by induction.

	Note that $v_1\alp$ is not a proper prefix of $u_1\alp$;
	otherwise, $v_1$ would be a $\alp$-prefix of $w$ which is longer than $u_0=\e$ but shorter than $u_1$.
	Symmetrically, $u_1\alp$ is not a proper prefix of $v_1\alp$.
	By the first statement, we have $\Has(w\ov{u_i}) \sse u_1\alp\Sig^* \cap \Sig^*\ov\alp\ov{u_1}$ and
	$\Has(v_1 w) \sse v_1\alp\Sig^* \cap \Sig^*\ov\alp\ov{v_1}$ for all $1\leq i < m$ and $1\leq j < n$.
	Therefore, if $u_1 \neq v_1$,
	the intersection $\Has(w\ov{u_i})\cap\Has(v_j w)$ has to be empty
	which proves statement~2.

	For statement~3, assume first $u_j >_s u_i$.
	By Lemma~\ref{lem:karl} we obtain $u_j = u_{j-i}u_i$ and hence
	$w\ov{u_j} = w\ov{u_i}\ov{u_{j-i}} \in \Has(w \ov{u_i})$.
	We conclude $\Has(w \ov{u_j}) \sse \Has(w \ov{u_i})$.

	Now assume $u_i$ is not a suffix of $u_j$.
	As $w\ov{u_i}$ (resp.\ $w\ov{u_j}$) is a factor of every word in $\Has(w\ov{u_i})$
	(resp.\ $\Has(w\ov{u_j})$) and
	there is only one factor $w$ in every word from $\Has(w)$,
	it is plain that the intersection $\Has(w\ov{u_i}) \cap \Has(w\ov{u_j})$
	is empty (see Corollary~\ref{cor:nonx_initial_once}).
\end{proof}

Let us define the index sets
\begin{align*}
	I &= \set{i}{1\leq i < m \land u_i\notin\oneset{u_1,\ldots,u_{i-1}}^*}, \\
	J &= \set{j}{1\leq j < n \land v_j\notin\oneset{v_1,\ldots,v_{j-1}}^*}.
\end{align*}
Thus, for all $i\in I$, no proper suffix of $u_i$ belongs to $\Pa(w)$ and
for all $j\in J$, no proper prefix of $\ov{v_j}$ belongs to $\Sa(w)$, 
see Lemma~\ref{lem:karl}.
By the previous lemma, if $v_1\neq u_1$, then
$\Has(w)$ is the disjoint union
\begin{equation*}
	\Has(w) = \oneset{w} \cup \bigcup_{i\in I}\Has(w\ov{u_i}) \cup 
	\bigcup_{j\in J} \Has(v_jw).
\end{equation*}

Note that for every word $w\ov{u_i}$ with $i\in I$
or $v_jw$ with $j\in J$,
the shortest $\alp$-prefix is complementary to the shortest $\ov\alp$-suffix.
This observation leads us to an important theorem that allows us
to reduce the general case to the special case where $u_1 = v_1$.

\begin{theorem}\label{thm:uequalv}
	Let $w\in\aSSa$ be a non-crossing $(m,n)$-word with $m,n\geq 2$,
	let $\Pa(w) = \oneset{u_0,\ldots,u_{m-1}}$ such that $u_0 <_p \cdots <_p u_{m-1}$,
	let $\Sa(w) = \oneset{\ov{v_0},\ldots,\ov{v_{v-1}}}$ such that
	$\ov{v_{n-1}} >_s \cdots >_s \ov{v_0}$, and define $I$, $J$ as above.

	For $u_1\neq v_1$,
	$\Has(w)$ is regular if and only if
	$\Has(w\ov{u_i})$ is regular for all $i\in I$ and
	$\Has(v_jw)$  is regular for all $j\in J$.
\end{theorem}

\begin{proof}
	As we already stated
	\begin{equation*}
		\Has(w) = \oneset{w} \cup \bigcup_{i\in I}\Has(w\ov{u_i}) \cup 
		\bigcup_{j\in J} \Has(v_jw).
	\end{equation*}

	If every language in the union is regular, then the union itself is regular.
	Conversely, assume $\Has(w\ov{u_i})$ is not regular for some $i\in I$,
	then the intersection
	\begin{equation*}
		\Has(w) \cap( u_1\alp\Sig^* \cap
		\Sig^*\ov\alp\ov{u_1}) \cap \Sig^*w\ov{u_i}\Sig^* = \Has(w\ov{u_i})
	\end{equation*}
	is not regular and hence $\Has(w)$ is not regular either,
	by Lemma~\ref{lem:heinz}.
	The argument for non-regular $\Has(v_jw)$ with $j\in J$ is analogous.
\end{proof}

Theorem~\ref{thm:uequalv} justifies the assumption $u_1 = v_1$ that we make from now on.
The next two theorems prove a necessary and sufficient condition for the
regularity of $\Has(w)$.
We start by proving that the condition is sufficient.

\begin{theorem}\label{thm:sufficient}
	Let $w\in\aSSa$ be a non-crossing $(m,n)$-word with $m,n\geq 2$,
	let $\Pa(w) = \oneset{u_0,\ldots,u_{m-1}}$ such that $u_0 <_p \cdots <_p u_{m-1}$,
	and let $\Sa(w) = \oneset{\ov{v_0},\ldots,\ov{v_{v-1}}}$ such that
	$\ov{v_{n-1}} >_s \cdots >_s \ov{v_0}$.
	
	$\cH_\alpha^*(w)$ is regular if both of the following conditions hold: 
	\begin{enumerate}
	\item	for all $1 \le s < m$, either $u_s \in \ov{\Sa(w)}^*$ or
		$\ov{\Sa(w)}^* \subseteq \{u_1, \ldots, u_{s}\}^*$, and 
	\item	for all $1 \le t < n$, either $v_t \in P_\alpha(w)^*$ or
		$P_\alpha(w) \subseteq \{v_1, \ldots, v_{t}\}^*$t. 
	\end{enumerate}
\end{theorem}

\begin{proof}
	Assume that both the conditions 1 and 2 are satisfied. 
	We may assume that there is $1 \le s < m$ such that $u_s \not\in \ov{\Sa(w)}^*$
	or there is $1 \le t < n$ such that $v_t \not\in P_\alpha(w)^*$;
	otherwise Corollary~\ref{cor:sufone} yields the regularity. 
	In addition, we cannot assume the existence of both such $s$ and $t$ as
	$u_s \not\in \ov{\Sa(w)}^*$ implies $\ov{\Sa(w)}^* \subseteq \{u_1, \ldots, u_{s}\}^*$
	due to the condition 1. 
	The symmetry in the roles of conditions 1 and 2 enables us to assume that such $s$
	exists without loss of generality, and moreover,
	we can assume that for all $1 \le i < s$, $u_i \in \ov{\Sa(w)}^*$
	and for all $s \le i < m$, $u_i \not\in \ov{\Sa(w)}^*$ by Lemma~\ref{lem:inclusions}. 

	Let $R' = \ov{\Sa(w)}^* w \{\ov{u_1}, \ldots, \ov{u_{s-1}}\}^*$ and for $s \le i < m$, let
	\begin{equation*}
		R_i = \{u_1, \ldots, u_i\}^* w \{\ov{u_1}, \ldots, \ov{u_i}\}^*
			\ov{u_i} \{\ov{u_1}, \ldots, \ov{u_i}\}^*.
	\end{equation*}
	Then we let $R = \bigcup_{s \le i < m} R_i \cup R'$ and we claim $\Has(w) = R$. 

	Firstly, we prove that $R \subseteq \cH_\alpha^*(w)$. 
	Since $m, n \ge 2$, Corollary~\ref{cor:1st_step_mn_ge2} can be used to see that
	$w \{\ov{u_1}, \ldots, \ov{u_i}\}^*\ov{u_i} \subseteq \cH_\alpha^*(w)$ for
	$s \le i < m$, and by Lemma~\ref{lem:theo}, 
	\begin{equation*}
		R_i = \{u_1, \ldots, u_i\}^* w \{\ov{u_1}, \ldots, \ov{u_i}\}^*
			\ov{u_i} \{\ov{u_1}, \ldots, \ov{u_i}\}^* \subseteq \cH_\alpha^*(w). 
	\end{equation*}

	Consider $z\in R'$.
	We may factorize $z = x_i\cdots x_1 w \ov{y_1}\cdots \ov{y_j}$ where
	$x_1,\ldots,x_i\in \ov{\Sa(w)}$ and $y_1,\ldots,y_j \in \oneset{u_1,\ldots,u_{s-1}}$.
	Let $1\leq t < n$ be minimal such that $v_t\notin\oneset{u_1,\ldots,u_{s-1}}^*$.
	As $v_t\in\oneset{u_1,\ldots,u_s}^*$, we see that $u_s$ is a factor of $v_t$ and $t\geq s$.
	By Lemma~\ref{lem:siegbert}, $u_{s-1}\in\oneset{v_1,\ldots,v_{t-1}}^*$
	and, by the minimality of $t$, $v_{t-1}\in\oneset{u_1,\ldots,u_{s-1}}^*$.
	If $\ia(x_\ell) < t$ for all $1\leq\ell\leq i$,
	then, due to Lemma~\ref{lem:theo},
	\begin{equation*}
		z\in \oneset{v_1,\ldots,v_{t-1}}^* w \oneset{\ov{v_1},\ldots,\ov{v_{t-1}}}^*\sse\Has(w).
	\end{equation*}
	Otherwise, let  $1\leq \ell \leq i$ be maximal such that $\ia(x_\ell) \ge t$
	and let $w' = x_\ell \cdots x_1 w$.
	Observe that $w\to_\LH^* w'$ and, again by Lemma~\ref{lem:theo},
	\begin{equation*}
		z \in \oneset{v_1,\ldots,v_{t-1}}^*w'\oneset{\ov{v_1},\ldots,\ov{v_{t-1}}}^*
			\sse\Has(w')\sse \Has(w).
	\end{equation*}
	
	Thus, $R \subseteq \cH_\alpha^*(w)$. 

	Now we prove the opposite inclusion by induction on the length of a derivation
	to generate a word in $\cH_\alpha^*(w)$ from $w$. 
	Clearly, $w \in R$ (base case). 
	As an induction hypothesis, we assume that any word which can be derived from $w$ by
	$\ell-1$ hairpin completions is in $R$, and consider $z \in \cH_\alpha^*(w)$
	whose shortest derivation from $w$ by hairpin completions is of length $\ell$. 
	Let $w'$ be the word that precedes $z$ on this shortest derivation, that is,
	$w' \in \Ha^{\ell-1}(w)$; hence, $w' \in R$ by the induction hypothesis. 
	Therefore, $w'$ must be either in $R_i$ for some $s \le i < m$ or in $R'$. 
	Let us consider the first case first. 
	If $z$ is obtained from $w'$ by right hairpin completion, then the complement of the
	extended part is in $\{u_1, \ldots, u_i\}^*\{\e, u_1, \ldots, u_{m-1}\}$,
	and hence, $z \in R_j$ for some $i \le j < m$. 
	Otherwise ($w' \to_\LH z$), $z \in R_i$ because
	$\ov{\Sa(w')} \subseteq \oneset{u_1, \ldots, u_s}^* \sse \oneset{u_1,\ldots,u_i}^*$. 
	Next we consider the case when $w' \in R'$. 
	Since $\oneset{u_1,\ldots,u_{s-1}}^* \sse \ov{\Sa(w)}^*$ it is easy to see that if $w' \to_\LH z$,
	then $z \in R'$ as well.
	Otherwise ($w' \to_\RH z$), $z \in w'\oneset{\e,\ov{u_1},\ldots,\ov{u_{m-1}}}\Sa(w)^*$ and,
	as $\Sa(w)^* \sse \oneset{\ov{u_1},\ldots,\ov{u_s}}^*$,
	if $z\notin R'$, this word is covered by some language $R_i$ where $s\leq i< m$.
	
	Consequently, $\cH_\alpha^*(w) = R$ is regular. \qed
\end{proof}

\begin{theorem}\label{thm:necessary}
	Let $w\in\aSSa$ be a non-crossing $(m,n)$-word with $m,n\geq 2$,
	let $\Pa(w) = \oneset{u_0,\ldots,u_{m-1}}$ such that $u_0 <_p \cdots <_p u_{m-1}$,
	let $\Sa(w) = \oneset{\ov{v_0},\ldots,\ov{v_{v-1}}}$ such that
	$\ov{v_{n-1}} >_s \cdots >_s \ov{v_0}$,
	and let $u_1 = v_1$.
	\begin{enumerate}
		\item
			$\Has(w)$ is not regular if there are $1\leq s < m$ and $1\leq t < n$
			such that $u_s\notin \oneset{v_1,\ldots,v_{n-1}}^*$ and
			$v_t\notin \oneset{u_1,\ldots,u_s}^*$.
		\item
			$\Has(w)$ is not regular if there are $1\leq s < m$ and $1\leq t < n$
			such that $u_s\notin \oneset{v_1,\ldots,v_{t}}^*$ and
			$v_t\notin \oneset{u_1,\ldots,u_{m-1}}^*$.
	\end{enumerate}
\end{theorem}

\begin{proof}
	Let $s$ and $t$ be the minimal indices that satisfy the conditions in statement~1.
	Note that $s,t \geq 2$ and $u_s,v_t \notin u_1^+$ as $u_1 = v_1$.
	We will first argue, why the assumption $s\leq t$ is no restriction.

	Let us consider the case where the conditions in statement~1 are satisfied,
	but the conditions in statement~2 are not satisfied.
	It is easy to see that $u_s\notin \oneset{v_1,\ldots,v_t}^*$ is satisfied anyway.
	By contradiction assume $s > t$.
	Due to Lemma~\ref{lem:siegbert},
	\begin{align*}
		v_t \notin \oneset{u_1,\ldots,u_s}^* \quad&\iff\quad
		v_t \notin \oneset{u_1,\ldots,u_t}^* \\
		&\iff\quad
		v_t \notin \oneset{u_1,\ldots,u_{m-1}}^*.
	\end{align*}
	This satisfies the conditions of statement~2 and yields the contradiction.
	We conclude $s \leq t$.

	Now, let us consider the case where the conditions of both statements are satisfied.
	Let $s$ and $t'$ be the minimal indices such that
	$u_s \notin \oneset{v_1,\ldots,v_{n-1}}^*$ and $v_{t'}\notin\oneset{u_1,\ldots,u_{m-1}}^*$.
	We may assume $s \leq t'$, by symmetry.
	Note that $v_{t'-1} \in\oneset{u_1,\ldots,u_{m-1}}^*$, by the minimality of $t'$.
	If $v_{t'-1} \in\oneset{u_1,\ldots,v_{s}}^*$, then we see that $s$ and $t = t'$
	are the minimal indices that satisfy the conditions in statement~1 and $s\leq t$.
	Otherwise, there is a factorization
	$v_{t'-1} = x u_i y$ where $x\in \oneset{u_1,\ldots, u_s}^*$, $s < i < m$,
	and $y\in \oneset{u_1,\ldots,u_{m-1}}^*$.
	Note that $s$ and $t = \ia(x)+s+1$ (hence $v_t = xu_{s+1}$)
	are the minimal indices that satisfy the conditions
	of statement~1 and, obviously, $s\leq t$.
	
	Observe that the minimality of $s$ yields
	$u_1,\ldots, u_{s-1} \in \oneset{v_1,\ldots, v_{n-1}}^*$.
	If $x \in \oneset{u_1,\ldots,u_{s-1},v_1,\ldots,v_{n-1}}$
	was a suffix of $u_s$, then $u_s = u_{s-\ia(x)}x \in \oneset{v_1,\ldots, v_{n-1}}^*$,
	due to Lemma~\ref{lem:karl}.
	Hence, none of these words is a suffix of $u_s$.
	Symmetrically, none of the words $u_1,\ldots,u_{s}$, $v_1,\ldots,v_{t-1}$
	is a suffix of $v_t$.
	This observation will become crucial later.
	
	We will now define a regular language $R$ and show that the intersection
	$\Has(w) \cap R$ is not regular and, therefore, $\Has(w)$ is not regular either.
	We let
	\begin{equation*}
		R = u_s u_1^{\geq n} v_t w \ov{u_1}^{\geq n} \ov{u_s}
	\end{equation*}
	and we claim
	\begin{equation*}
		\Has(w) \cap R = \set{u_su_1^\ell v_t w \ov{u_1}^\ell \ov{u_s}}
			{\ell \geq n} =: L,
	\end{equation*}
	which is obviously not regular.
	Note that for every $\ell\geq n$
	\begin{equation*}
		w \ra_\RH^* w\ov{u_1}^\ell \ra_\RH w\ov{u_1}^\ell \ov{u_s} \ra_\LH
			u_su_1^\ell v_t w \ov{u_1}^\ell \ov{u_s}.
	\end{equation*}
	Hence, $\Has(w)\cap R \supseteq L$.
	
	Let $z = u_su_1^{\ell_1} v_t w \ov{u_1}^{\ell_2} \ov{u_s}$
	for some $\ell_1,\ell_2 \geq n$, which is in $R$.
	We assume $z\in \Has(w)$ and prove that this assumption requires $\ell_1 = \ell_2$.
	Let $z'$ be the {\em right-most} word in the derivation
	$w \ra_\cH^* z' \ra_\cH^* z$ such that $z' = x w y$
	for some words $x,y$ with $u_1^{\ell_1} v_t \geq_s x$ and $y \leq_p \ov{u_1}^{\ell_2}$;
	these conditions  mean that $x$ or $y$ does not
	overlap with the prefix $u_s$ or the suffix $\ov{u_s}$, respectively.
	By right-most we mean that either $z'\ra_\LH x'z'\ra_\cH^* z$
	where $x'x>_s u_1^{\ell_1}v_t$ or 
	$z'\ra_\RH z'y' \ra_\cH^* z$ where $\ov{u_1}^{\ell_2} <_p yy'$;
	this means $x'$ or $y'$ overlaps with the prefix $u_s$ or the suffix $\ov{u_s}$,
	respectively.
	Obviously, $y\in \ov{u_1}^*$.
	Note that if $x\neq \e$, then $x$ cannot be a proper suffix of $v_t$; otherwise
	a word from $u_1 = v_1,\ldots,v_{t-1}$ would be a suffix of $v_t$ which was excluded.
	Hence, $x = \e$ or $x\in u_1^*v_t$.

	First consider the case $z'\ra_\LH x'z'\ra_\cH^* z$ where $x'x>_s u_1^{\ell_1}v_t$.
	We show that this case cannot occur.
	Let $u'\neq \e$ be the suffix of $u_s$ such that $u'u_1^{\ell_1}v_t = x'x$.
	As $x'\in u_1^*\oneset{v_1,\ldots,v_{n-1}}$ and $u'\alp \leq_p x'\alp$,
	some word from $u_1=v_1,\ldots,v_{n-1}$ would be a suffix of $u_s$.
	
	Now consider the case $z'\ra_\RH z'y' \ra_\cH^* z$ where $\ov{u_1}^{\ell_2} <_p yy'$.
	Again, let $u'\neq \e$ be the suffix of $u_s$ such that $\ov{u_1}^{\ell_2}\ov{u'}= yy'$.
	As $u'\alp$ is a prefix of $xu_{m-1}\alp$ and none of the words $v_1,\ldots,v_t$,
	$u_1,\ldots,u_{s-1}$ is a suffix of $u_s$, we see that $u' = u_s$ and $x = \e$.
	
	Thus, in order to generate $z$ from $w$ by iterated hairpin completion, the derivation
	process must be of the form 
	\begin{equation*}
		w \ra_\RH^* w\ov{u_1}^{\ell_2}\ov{u_s} \ra_\LH^*
			u_s u_1^{\ell_1} v_t w\ov{u_1}^{\ell_2}\ov{u_s} = z.
	\end{equation*}
	
	Let $x$ be a (newly chosen) word
	such that $w\ov{u_1}^{\ell_2}\ov{u_s} \ra_\LH xw\ov{u_1}^{\ell_2}\ov{u_s}$
	is the first left hairpin completion in the derivation above.
	Therefore, $x\alp$ is a prefix of $u_s u_1^{\ell_2} v_{n-1}\alp$
	and $x$ is a suffix of $u_s u_1^{\ell_1} v_t$.
	In particular, every suffix $y$ of $x$ with $\ia(y)\leq t$ is a suffix of $v_t$.
	Recall that $\ia(u_s) = s\leq t$.
	If $x\alp$ was a prefix of $u_s\alp$, then
	some word from $u_1,\ldots, u_s$ would be a suffix of $v_t$ which is impossible.
	Verify that $x\in u_su_1^+$ and $x = u_su_1^{\ell_2}v_j$ with $1\leq j < t$
	would also impose a forbidden suffix for $v_t$.
	Thus, we see that $x = u_su_1^{\ell_2}v_j$ with $t\leq j < n$.
	The case $j > t$ is not possible as it implies
	$v_t\alp <_p v_j\alp = u_1^{j-t}v_t\alp$ and a word from 
	$u_1 = v_1,\ldots,v_{t-1}$ would be a suffix of $v_t$.
	Therefore, $x = u_s u_1^{\ell_2} v_t$ and since $x$ is a suffix of
	$u_s u_1^{\ell_1} v_t$ and $u_1$ is not a suffix of $u_s$
	we deduce $u_s u_1^{\ell_2} v_t = u_s u_1^{\ell_1} v_t$.	
	
	Consequentially, $z\in\Has(w)$ if and only if $\ell_1 = \ell_2$.
	This completes the proof of $\Has(w)\cap R = L$. \qed
\end{proof}

Combining Theorems~\ref{thm:sufficient} and~\ref{thm:necessary}, we conclude:

\begin{corollary}\label{cor:condition}
	Let $w\in\aSSa$ be a non-crossing $(m,n)$-word with $m,n\geq 2$,
	let $\Pa(w) = \oneset{u_0,\ldots,u_{m-1}}$ such that $u_0 <_p \cdots <_p u_{m-1}$,
	let $\Sa(w) = \oneset{\ov{v_0},\ldots,\ov{v_{v-1}}}$ such that
	$\ov{v_{n-1}} >_s \cdots >_s \ov{v_0}$,
	and let $u_1 = v_1$.

	$\Has(w)$ is regular if and only if
	\begin{enumerate}
		\item for all $1\leq s < m$ we have $u_s \in \ov{\Sa(w)}^*$
			or $\ov{\Sa(w)}\sse \oneset{u_1,\ldots,u_s}^*$ and
		\item for all $1\leq t < n$ we have $v_t \in \Pa(w)^*$
			or $\Pa(w)\sse \oneset{v_1,\ldots,v_t}^*$.
	\end{enumerate}
\end{corollary}

\begin{proof}
	Theorem~\ref{thm:sufficient} provides the if-part.
	
	For the only-if-part, assume that condition 1 is breached.
	There exists $1\leq s < m$ such that $u_s\notin\ov{\Sa(w)}^*$
	and $\ov{\Sa(w)}\setminus \oneset{u_1,\ldots,u_s}^* \neq \es$.
	Let $1\leq t < n$ such that $v_t \in \ov{\Sa(w)}\setminus \oneset{u_1,\ldots,u_s}^*$.
	We see that $s,t$ are indices satisfying the conditions of statement 1 in
	Theorem~\ref{thm:necessary} and, therefore, $\Has(w)$ is not regular.
\end{proof}

Thus, we provided a necessary and sufficient condition for the regularity
of a non-crossing $(m,n)$-word.
As one can easily observe, this condition is decidable.

\begin{corollary}\label{cor:decidable}
	For a given non-$\alp$-crossing word $w$,
	it is decidable whether or not its iterated hairpin completion, $\Has(w)$, is regular.
\end{corollary}

Furthermore, we can derive from the proof of Theorem~\ref{thm:necessary}
that if the iterated hairpin completion of $w$ is not regular, then
the intersection of $\Has(w)$ with $R = u_su_1^{\ge n} v_t w \ov{u_1}^{\ge n} \ov{u_s}$
is not a regular language (for suitable $s,t$ and in case $u_1 = v_1$).
More precisely, we obtained the context-free language
\begin{equation*}
	\Has(w) \cap R = \set{u_su_1^\ell v_t w \ov{u_1}^\ell \ov{u_s}}
		{\ell \geq n}.
\end{equation*}

Consider we intersect $\Has(w)$ with $R' = (u_su_1^{\ge n} v_t)^2 w \ov{u_1}^{\ge n} \ov{u_s}$.
Using the same arguments as we did within the proof of Theorem~\ref{thm:necessary},
we can show that
\begin{equation*}
	\Has(w) \cap R' = \set{u_su_1^\ell v_tu_su_1^\ell v_t w \ov{u_1}^\ell \ov{u_s}}
		{\ell\geq n},
\end{equation*}
which is a non-context-free language.
Using this idea we can proof that if $\Has(w)$ is not regular, it is not context-free either.
The details of this proof are left for the interested reader.

\begin{corollary}\label{cor:noncf}
	Let $w$ be a non-$\alp$-crossing word.
	If its iterated hairpin completion $\Has(w)$ is not regular,
	then $\Has(w)$ is not context-free.
\end{corollary}

\section*{Final Remarks}

We prove that regularity of iterated hairpin completion a given of non-crossing
word is decidable.
The general case, including that of crossing words, remains to be explored.


\newcommand{\Ju}{Ju}\newcommand{\Ph}{Ph}\newcommand{\Th}{Th}\newcommand{\Ch}{C%
h}\newcommand{\Yu}{Yu}\newcommand{\Zh}{Zh}


\begin{thebibliography}{10}
\providecommand{\url}[1]{\texttt{#1}}
\providecommand{\urlprefix}{URL }

\bibitem{Adleman94}
Adleman, L.M.: Molecular computation of solutions to combinatorial problems.
  Science  266(5187),  1021--1024 (1994)

\bibitem{AritaKobayashi02}
Arita, M., Kobayashi, S.: {DNA} sequence design using templates. New Generation
  Computing  20,  263--277 (2002)

\bibitem{ChepteaMM06}
Cheptea, D., Mart\'{\i}n-Vide, C., Mitrana, V.: A new operation on words
  suggested by {DNA} biochemistry: Hairpin completion. Transgressive Computing
  pp. 216--228 (2006)

\bibitem{DieKop10}
Diekert, V., Kopecki, S.: Complexity results and the growths of hairpin
  completions of regular languages (extended abstract). In: Domaratzki, M.,
  Salomaa, K. (eds.) CIAA. Lecture Notes in Computer Science, vol. 6482, pp.
  105--114. Springer (2010)

\bibitem{Hagiya97}
Hagiya, M., Arita, M., Kiga, D., Sakamoto, K., Yokoyama, S.: Towards parallel
  evaluation and learning of boolean $\mu$-formulas with molecules. In: Second
  Annual Genetic Programming Conf. pp. 105--114 (1997)

\bibitem{HU}
{Hopcroft}, J.E., {Ulman}, J.D.: Introduction to Automata Theory, Languages and
  Computation. Addison-Wesley (1979)

\bibitem{JoKeMa02}
Jonoska, N., Kephart, D., Mahalingam, K.: Generating {DNA} codewords.
  Congressus Numerantium  156,  99--110 (2002)

\bibitem{JonoskaMahalingam04}
Jonoska, N., Mahalingam, K.: Languages of {DNA} based code words. In: Chen, J.,
  Reif, J. (eds.) Proc.~of DNA Computing 9. Lecture Notes in Computer Science,
  vol. LNCS 2943, pp. 61--73. Springer (2004, 61-73)

\bibitem{KaYaOhYaHa08}
Kameda, A., Yamamoto, M., Ohuchi, A., Yaegashi, S., Hagiya, M.: Unravel four
  hairpins! Natural Computing  7,  287--298 (2008)

\bibitem{KaKoLoSoTh06}
Kari, L., Konstantinidis, S., Losseva, E., Sos\'{i}k, P., Thierrin, G.: A
  formal language analysis of {DNA} hairpin structures. Fundamenta Informaticae
   71(4),  453--475 (2006)

\bibitem{KariKopeckiSeki}
Kari, L., Kopecki, S., Seki, S.: On the regularity of iterated hairpin
  completion of a single word. Fundamenta Informaticae  (to appear) (2011),
  also in CoRR abs/1104.2385

\bibitem{Kopecki11}
Kopecki, S.: On iterated hairpin completion. Theoretical Computer Science
  412(29),  3629--3638 (July 2011)

\bibitem{Manea10}
Manea, F.: A series of algorithmic results related to the iterated hairpin
  completion. Theor. Comput. Sci.  411(48),  4162--4178 (2010)

\bibitem{ManeaMM09}
Manea, F., Mart\'{\i}n-Vide, C., Mitrana, V.: On some algorithmic problems
  regarding the hairpin completion. Discrete Applied Mathematics  157(9),
  2143--2152 (2009)

\bibitem{DBLP:conf/cie/ManeaM07}
Manea, F., Mitrana, V.: Hairpin completion versus hairpin reduction. In:
  Cooper, S.B., L{\"o}we, B., Sorbi, A. (eds.) CiE. Lecture Notes in Computer
  Science, vol. 4497, pp. 532--541. Springer (2007)

\bibitem{ManeaMY09tcs}
Manea, F., Mitrana, V., Yokomori, T.: Two complementary operations inspired by
  the {DNA} hairpin formation: Completion and reduction. Theor. Comput. Sci.
  410(4-5),  417--425 (2009)

\bibitem{ManeaMY10}
Manea, F., Mitrana, V., Yokomori, T.: Some remarks on the hairpin completion.
  Int. J. Found. Comput. Sci.  21(5),  859--872 (2010)

\bibitem{PaRoYo01}
P\u{a}un, G., Rozenberg, G., Yokomori, T.: Hairpin languages. International
  Journal of Foundations of Computer Science  12(6),  837--847 (2001)

\bibitem{Sakamoto98}
Sakamoto, K., Kiga, D., Komiya, K., Gouzu, H., Yokoyama, S., Ikeda, S., Hagiya,
  M.: State transitions by molecules (1998)

\bibitem{TakinoueSuyama04}
Takinoue, M., Suyama, A.: Molecular reactions for a molecular memory based on
  hairpin {DNA}. Chem-Bio Informatics Journal  4,  93--100 (2004)

\bibitem{TakinoueSuyama06}
Takinoue, M., Suyama, A.: Hairpin-{DNA} memory using molecular addressing.
  Small  2(11),  1244--1247 (2006)

\bibitem{Winfree98}
Winfree, E.: Whiplash {PCR} for {O(1)} computing. In: University of
  Pennsylvania. pp. 175--188 (1998)

\end{thebibliography}
\end{document}